\documentclass[11pt,a4paper]{article}
\usepackage[T1]{fontenc}
\usepackage[latin9]{inputenc}
\usepackage{amsmath}
\usepackage{amssymb}
\usepackage[boxed,boxruled,vlined]{algorithm2e}
\usepackage[unicode=true,pdfusetitle,
 bookmarks=true,bookmarksnumbered=false,bookmarksopen=false,
 breaklinks=false,pdfborder={0 0 1},backref=section,colorlinks=false]
 {hyperref}

\makeatletter
 
\usepackage{xspace}
\usepackage{lineno}
\usepackage{fullpage}
\usepackage{marvosym}
\usepackage{enumerate}
\usepackage{dsfont}
\usepackage{multirow}
\usepackage[table]{xcolor}
\usepackage[nofancy]{svninfo}
\usepackage[boxruled,vlined]{algorithm2e}
\usepackage[boldsans]{concmath}
\usepackage{euler}
\usepackage{amsfonts}
\usepackage{dsfont}
\usepackage{bm}
\usepackage{amsthm}
\usepackage{wasysym}
\usepackage{graphicx}

\newtheorem{lemma}[equation]{Lemma}

\newtheorem{proposition}[equation]{Proposition}
\newtheorem{theorem}[equation]{Theorem}

\usepackage{thmtools}
\usepackage{thm-restate}

\newcommand{\D}{\ensuremath{\mathcal{D}}}
\newcommand{\Hy}{\ensuremath{\mathcal{H}}}

\newcommand{\N}{\ensuremath{\mathds N}}

\newcommand{\R}{\ensuremath{\mathds R}}

\newcommand{\HRule}{\rule{\linewidth}{0.5mm}}
\newcommand{\remove}[1]{}
\renewcommand{\epsilon}{\varepsilon}
\makeatother

\begin{document}
\begin{titlepage}

\begin{center}
\HRule \\[0.4cm] {\huge Halving Balls in Deterministic Linear Time}\\[0cm]
\HRule \\[1cm]
\end{center}

\begin{minipage}[c]{0.3\textwidth}%
  \begin{center}
  \textsc{Michael Hoffmann}\\
  {\small Institute of Theoretical Computer Science,}\\
  {\small ETH Z\"urich, Switzerland \\[0.4cm] }
  \par\end{center}%
\end{minipage}%
\begin{minipage}[c]{0.3\textwidth}%
  \begin{center}
  \textsc{Vincent Kusters}\\
  {\small Institute of Theoretical Computer Science,}\\
  {\small ETH Z\"urich, Switzerland \\[0.4cm] }
  \par\end{center}%
\end{minipage}
\begin{minipage}[c]{0.3\textwidth}%
  \begin{center}
  \textsc{Tillmann Miltzow}\\
  {\small Institute of Computer Science,}\\
  {\small Freie Universit\"at Berlin, Germany \\[0.4cm] }
  \par\end{center}%
\end{minipage}\rule{1\linewidth}{0mm} \\[0.6cm]

\begin{abstract}
  Let $\D$ be a set of $n$ pairwise disjoint unit balls in $\R^d$ and
  $P$ the set of their center points. A hyperplane $\Hy$ is an
  \emph{$m$-separator} for $\D$ if each closed halfspace bounded by
  $\Hy$ contains at least $m$ points from $P$. This generalizes the
  notion of halving hyperplanes, which correspond to
  $n/2$-separators. The analogous notion for point sets has been well
  studied. Separators have various applications, for instance, in
  divide-and-conquer schemes. In such a scheme any ball that is
  intersected by the separating hyperplane may still interact with
  both sides of the partition. Therefore it is desirable that the
  separating hyperplane intersects a small number of balls only.
  
  We present three deterministic algorithms to bisect or approximately
  bisect a given set of disjoint unit balls by a hyperplane: Firstly,
  we present a simple linear-time algorithm to construct an $\alpha
  n$-separator for balls in $\R^d$, for any $0<\alpha<1/2$, that
  intersects at most $cn^{(d-1)/d}$ balls, for some constant $c$ that
  depends on $d$ and $\alpha$. The number of intersected balls is best
  possible up to the constant $c$.  Secondly, we present a near-linear
  time algorithm to construct an $(n/2-o(n))$-separator in $\R^d$ that
  intersects $o(n)$ balls.  Finally, we give a linear-time algorithm
  to construct a halving line in $\R^2$ that intersects
  $O(n^{(5/6)+\epsilon})$ disks.

  Our results improve the runtime of a disk sliding algorithm by
  Bereg, Dumitrescu and Pach. In addition, our results improve and
  derandomize an algorithm to construct a space decomposition used by
  L{\"o}ffler and Mulzer to construct an onion (convex layer)
  decomposition for imprecise points (any point resides at an unknown
  location within a given disk).
\end{abstract}

\medskip

\begin{figure}[htbp]
  \centering\includegraphics[scale=0.81]{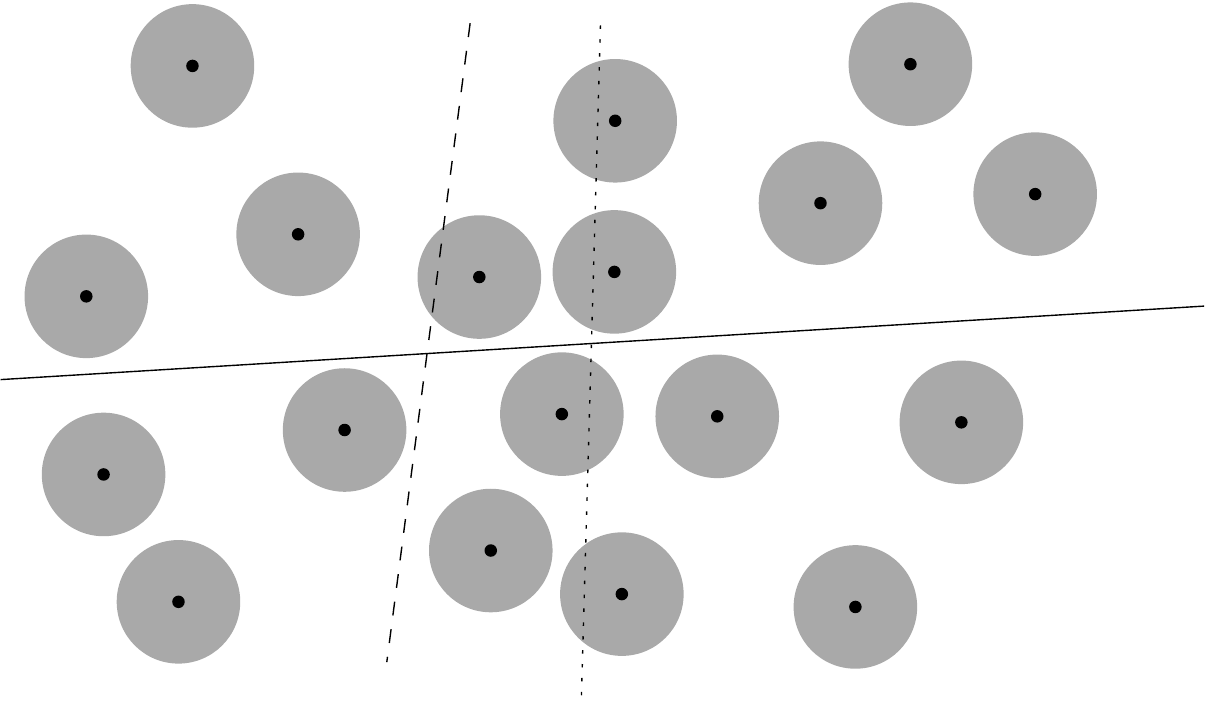}
  \caption{A set of $18$ disks in $\R^2$ and three separators. The
    dashed line forms a 6-separator. Both the solid line and the
    dotted line are halving lines. The solid line is preferable to the
    other two lines because it separates perfectly and intersects
    no disks.}
\end{figure}

%\begin{center}
%  \vfill{}
%  {\large Revision \svnInfoRevision\ --- \svnToday}
%  \par\end{center}
%{\large \par}

\end{titlepage}
 
% ---------------------------- Text --------------------------------------
%\linenumbers%
\section{Introduction}
\label{sec:introduction}
Let $\D$ be a set of $n$ pairwise disjoint unit balls in $\R^d$ and
$P$ the set of their center points. A hyperplane $\Hy$ is an
\emph{$m$-separator} for $\D$ if each closed halfspace bounded by
$\Hy$ contains at least $m$ points from $P$. This generalizes the
notion of halving hyperplanes, which correspond to
$n/2$-separators. The analogous notion of separating hyperplanes for
point sets has been well studied (see, e.g,~\cite{martini1998median}
for a survey).  Separators have various applications, for instance in
divide-and-conquer schemes (we discuss some explicit examples
below). In such a scheme any ball that is intersected by the
separating hyperplane may still interact with both sides of the
partition.  Therefore it is desirable that the separating hyperplane
intersects a small number of balls only.

Alon, Katchalski and Pulleyblank~\cite{DBLP:journals/dcg/AlonKP89}
prove that for any set $\D$ in $\R^2$, there exists a direction such
that every line with this direction intersects $O(\sqrt{n\log n})$
disks. In particular, this guarantees the existence of a halving line
that intersects at most $O(\sqrt{n\log n})$ disks. L{\"o}ffler and
Mulzer~\cite{DBLP:journals/corr/abs-1302-5328} observed that this
proof gives a randomized linear-time algorithm. In this paper, we
present the following three deterministic algorithms, each of which
computes an $m$-separator that intersects $O(n)$ balls for various
$m$.
\begin{theorem}\label{thm:kconstant}
  Given a set $\D$ of $n$ pairwise disjoint unit balls in $\R^d$ and
  $\alpha\in(0,1/2)$, one can construct in $O((1-2\alpha)n)$ time a
  hyperplane $\Hy$ that intersects $O((n/(1-2\alpha))^{(d-1)/d})$
  balls from $\D$ and such that each closed halfspace bounded by $\Hy$
  contains at least $\alpha n$ centers of balls from $\D$. The
  constants hidden by the asymptotic notation depend on $d$ only.
\end{theorem}
\begin{theorem}\label{thm:kslow}
  Given a set $\D$ of $n$ pairwise disjoint unit balls in $\R^d$ and a
  function $f(n)\in \omega(1)\cap O(\log n)$, one can construct in
  $O(nf(n))$ time a hyperplane $\Hy$ such that each closed halfspace
  bounded by $\Hy$ contains at least
  $\frac{n}{2}(1-1/f(n))=\frac{n}{2}(1-o(1))$ balls from $\D$.
\end{theorem}
\begin{theorem}\label{thm:2dlinear}
  For any set $\D$ of $n$ pairwise disjoint unit disks in $\R^2$ and
  any $\varepsilon>0$ one can construct in $O(n)$ time a line $\ell$
  that intersects $O(n^{(5/6)+\varepsilon})$ disks from $\D$ and such
  that each closed halfplane bounded by $\ell$ contains at least $n/2$
  centers of disks from $\D$.
\end{theorem}
\noindent We develop a generic algorithm in $\R^d$ that can be
instantiated with different parameters to obtain
Theorem~\ref{thm:kconstant} and Theorem~\ref{thm:kslow}. Note that
Theorem~\ref{thm:kslow} improves the separation of the center points
(compared to Theorem~\ref{thm:kconstant}) at the cost of increasing
the running time slightly. Theorem~\ref{thm:2dlinear} computes a true
halving line in the plane.

\paragraph{Related work.} Bereg, Dumitrescu and
Pach~\cite{bereg2008sliding} (see also~\cite[Lemma
9.3.2]{pach2009combinatorial}) strengthen the initial result of Alon,
Katchalski and Pulleyblank slightly by proving that there exists a
direction such that any line with this direction has at most
$O(\sqrt{n \log n})$ disks \emph{within constant distance}. They use
this lemma to prove that one can always move a set of $n$ unit disks
from a start to a target configuration in $3n/2+O(\sqrt{n\log n})$
moves. Their algorithm runs in $O(n^{3/2}{(\log n)}^{-1/2})$ time,
which Theorem~\ref{thm:2dlinear} improves to $O(n\log
n)$.
%\todo{Reviewer 1 suggests we give details on how the improvement
%  works.}
%  TM: You can reformulate the problem as follows: Given a set of points in the plane with pairwise 
%  distance $2$ find a strip of width $2$ that intersects at most .... many points.
%  Now you can see that the problem doesn't change so much. Actually it
%  is the same proof, just changing some constants at certain places. I don't find it interesting.

Held and Mitchell~\cite{held2008triangulating} introduced a paradigm
for modeling data imprecision where the location of a point in the
plane is not known exactly. For each point, however, we are given a
unit disk that is guaranteed to contain the point. The authors show
that after preprocessing the disks in $O(n\log n)$ time, they can
construct a triangulation of the actual point set in linear time.
L{\"o}ffler and Mulzer~\cite{DBLP:journals/corr/abs-1302-5328} follow
the same model to construct the onion layer of an imprecise point set.
They observed that the proof by Alon~et~al.\ immediately gives a
randomized expected linear-time algorithm in the following fashion.
Pick an angle $\beta\in[0,\pi]$ uniformly at random and compute a
halving line for the disks with slope $\beta$. This halving line
intersects at most $O(\sqrt{n\log n})$ unit disks with probability at
least $1/2$. L{\"o}ffler and Mulzer use this algorithm to compute a
\emph{$(\alpha,\beta)$-space decomposition tree}: a data structure
similar to a binary space partition in which every line is an $\alpha
k$-separator that intersects at most $k^\beta$ disks. They show that
such a $(1/2 + \epsilon, 1/2 + \epsilon)$ space decomposition tree can
be computed in $O(n\log n)$ expected time, for every $\epsilon > 0$.
Theorem~\ref{thm:kconstant} can be used to improve this to $O(n\log
n)$ deterministic time. They also present a simple deterministic
linear-time algorithm that guarantees that at least $n/10$ of the
disks are completely on each side of some axis-parallel line. Next,
they describe a more sophisticated, deterministic $O(n\log n)$
algorithm to compute a line $\ell$ such that there are at least
$n/2-cn^{5/6}$ disks completely to each side of $\ell$. The algorithm
uses an $r$-partition of the plane~\cite{matouvsek1992efficient} to
find good candidate lines. Theorem~\ref{thm:2dlinear} can be used to
improve running time of this algorithm to $O(n)$.

Tverberg~\cite{tverberg1979seperation} studies a related question. He
proves that for every natural number $k$ there is a number $K(k)$,
such that given convex pairwise disjoint sets $C_1,\dots,C_{K(k)}$,
there always exists a line with some set completely on one side and
$k$ sets completely on the other side. Finally, the question has a
continuous counterpart that has been solved
recently~\cite{esposito2012longest}.

% XXX(TM: Shall we add some more literature. For instance that the claim
% by Pach after Lemma 9.3.2 is unproven, or already mention the
% Heilbronn problem here, or more on the continuous problem, or more on
% the back round for the Seperation into slabs and that it uses some
% pretty advanced algorithms? I have the sources added into the bib
% already.)

\paragraph{Organization.} We develop a generic algorithm to compute a
separator in $\R^d$ (where the trade-off between the number of
intersected disks and the number of disk centers on each side is
determined by a parameter) and prove Theorem~\ref{thm:kconstant} and
Theorem~\ref{thm:kslow} in
Section~\ref{sec:computing_a_separator_with_few_intersections}. We
prove Theorem~\ref{thm:2dlinear} in
Section~\ref{sec:computing_a_halving_line_with_few_intersections}. Our
algorithm follows the approach used in the linear-time ham-sandwich
cut algorithm~\cite{DBLP:journals/dcg/LoMS94}. It divides the line
arrangement dual to the set of disk center points by vertical lines
such that each slab (the region bounded by two consecutive vertical
lines) contains at most a constant fraction of the vertices of the
arrangement. In each iteration, the algorithm chooses a slab and
discards the rest of the arrangement.

\section{Separating balls in $\R^d$}
\label{sec:computing_a_separator_with_few_intersections}

In this section, we develop a generic algorithm to compute a separator
for a given set of pairwise disjoint unit balls in $\R^d$. Using this
generic algorithm, we will give two algorithms to compute an
approximately halving hyperplane that intersects a sublinear number of
balls.

Besides the set $\D$ of $n$ balls in $\R^d$, the generic algorithm has
two more parameters. First, a number $b\in\{1,\ldots,n\}$ that
quantifies the quality of the approximation: we will show that the
hyperplane constructed by the algorithm forms an $(n-b)/2$-separator
for $\D$. The main step of the algorithm consists in finding a
direction $d$ such that we are guaranteed to find a desired separator
that is orthogonal to $d$. A second parameter $k\in\N$ of the
algorithm specifies the number of different directions to generate and
test during this step. As a rule of thumb, generating more directions
results in a better solution, but the runtime of the algorithm
increases proportionally. The algorithm works for certain combinations
of these parameters only, as detailed in the following theorem.
\begin{restatable}{theorem}{thmApprox}\label{thm-approxPlane}
  Given a set $\D$ of $n$ pairwise disjoint unit balls in $\R^d$ and
  parameters $b\in\{1,\ldots,n\}$ and $k\in\N$ that satisfy the
  conditions
  \begin{eqnarray}
    d n & \leq & kb \quad\mathrm{and} \label{cond:t1}\\ 
    t :=
    \left({\frac{V_d}{2d^{(d-2)/2}}}\right)^{1/d}\frac{n^{1/d}}{k^{2-1/d}}
    & > & 2, \label{cond:t2}
  \end{eqnarray}
  (where $V_d$ is the volume of the $d$-dimensional unit ball), one
  can construct in $O(kn)$ time a hyperplane $\Hy$ that intersects at
  most $2b/(t-2)$ balls from $\D$ and such that each closed halfspace
  bounded by $\Hy$ contains at least $(n-b)/2$ centers of balls from
  $\D$.
\end{restatable}

Perhaps more interesting than Theorem~\ref{thm-approxPlane} in its
full generality are the special cases stated as
Theorem~\ref{thm:kconstant} and Theorem~\ref{thm:kslow} above.
Theorem~\ref{thm:kconstant} describes the case that $k$ is constant.
It can be obtained by choosing $b=\lfloor (1-2\alpha) n \rfloor$ and
$k=\lceil (1-2\alpha)d\rceil$ for $\alpha\in(0,1/2)$.
Theorem~\ref{thm:kslow} describes the case that $k$ is a very slowly
growing function $f(n)$. It can be obtained by choosing $b= n/f(n)$
and $k=df(n)$.
    
\paragraph{Overview of the algorithm.}
Our algorithm consists of two steps. In the first step, we find a
direction $d$ in which the balls from $\D$ are ``spread out
nicely''. More precisely, for an arbitrary (oriented) line $\ell$
consider the set $P$ of points that results from orthogonally
projecting all centers of balls from $\D$ onto $\ell$. Denote by
$p_1,\ldots,p_n$ the order of points from $P$ sorted along $\ell$. We
want to find an $(n-b)/2$-separator orthogonal to $\ell$. This means
that the separating hyperplane $\Hy$ must intersect $\ell$ somewhere
in between $p_{(n-b)/2}$ and $p_{(n+b)/2}$.

However, we also need to guarantee that not too many points from $P$
are within distance one of $\Hy$, which may or may not be possible
depending on the choice of $\ell$. Therefore we try several possible
directions/lines and select the first one among them that works. In
order to evaluate the quality of a line, we use as a simple criterion
the \emph{spread}, defined to be the distance between $p_{(n-b)/2}$
and $p_{(n+b)/2}$. Given a line $\ell$ with sufficient spread, we can
find a suitable $(n-b)/2$-separator orthogonal to $\ell$ in the second
step of our algorithm, as the following lemma demonstrates. Note the
safety cushion of width one to the remaining disks of $\D$.
\begin{lemma}\label{lem:Selectpoint}
  Given a set $P$ of $b$ (one-dimensional) points in an interval
  $[\ell,r]$ of length $w=r-\ell>2$, we can find in $O(b)$ time a
  point $p\in(\ell+1,r-1)$ such that at most $2b/(w-2)$ points from
  $P$ are within distance one of $p$.
\end{lemma}
\begin{proof}
  We select $\lceil(w-2)/2\rceil$ pairwise disjoint closed
  sub-intervals of length two in $(\ell,r)$. By the pigeonhole
  principle at least one these intervals contains at most
  $b/\lceil(w-2)/2\rceil\le 2b/(w-2)$ points from $P$. Select $p$ to
  be the midpoint of such an interval.

  Algorithmically, we can find such an interval using a kind of binary
  search on the intervals: We maintain a set of points and a range of
  intervals. At each step consider the median interval $I$ and test
  for every point whether it lies in $I$, to the left of $I$, or to
  the right of $I$. Then either $I$ contains at most $2b/(w-2)$ points
  from $P$ and we are done, or we recurse on the side that contains
  fewer points, after discarding all points and intervals on the other
  side. The process stops as soon as the current range of intervals
  contains at most $2b/(w-2)$ points from $P$, at which point any of
  the remaining intervals can be chosen. Given that we maintain the
  ratio between the number of points and the number of intervals, the
  process terminates with an interval of the desired type.  As the
  number of points decreases by a constant factor in each iteration,
  the overall number of comparisons can be bounded by a geometric
  series and the resulting runtime is linear.
\end{proof}

\paragraph{How to find a good direction.}
Our algorithm tries $k$ different directions and stops as soon as it
finds a direction with spread at least $t$ (see
Theorem~\ref{thm-approxPlane}). For a given direction the spread can
be computed in $O(n)$ time using linear time rank 
selection~\cite{blum1973time}. In the remainder of this section, we
will discuss how to select an appropriate set of directions such that
one direction is guaranteed to have spread at least $t$.

For this we need a bound on the number of balls simultaneously within
distance $w_1,\ldots, w_d$ of some hyperplanes $\Hy_1,\ldots,
\Hy_d$. Below we give an easy formula based on a volume argument. This
formula in turn motivates our choice of directions, which we will
explain thereafter.
\begin{lemma}\label{lem:intersection}
  Let $\vec{v}_1,\vec{v}_2,\ldots ,\vec{v}_d \in S^{d-1} \subset \R^d$
  be linearly independent directions and $\Hy_1, \Hy_2, \ldots ,\Hy_d$
  hyperplanes with corresponding normal directions, then the maximal
  number of pairwise disjoint unit balls entirely within distance
  $w_1,\ldots,w_d$ of $\Hy_1, \Hy_2, \ldots ,\Hy_d$, respectively, is
  bounded from above by $$ \frac{2^d w_1\ldots w_d }{| \det \left(
      \vec{v}_1,\ldots,\vec{v}_d \right)| V_d } ,$$ where $V_d$
  denotes the volume of the $d$-dimensional unit ball.
\end{lemma}
\begin{proof}
  For each hyperplane $\Hy_i$ consider the region $S_i$ within distance
  $w_i$ of $\Hy_i$. We want to count the number of balls in
  $S := \bigcap_i S_i$. As each ball has volume $V_d$ and
  they are pairwise disjoint, it is sufficient to bound the volume of
  $S$. The volume of $S$ depends linearly on $w_1,\ldots, w_d$, so we
  scale them all to one. We can map the linearly independent vectors
  $(\vec{v}_1,\ldots,\vec{v}_d)$ to the standard basis
  $(e_1,\ldots,e_d)$ by multiplying with the matrix
  $(\vec{v}_1,\ldots,\vec{v}_d)^{-1}$. The volume changes by this
  transformation by a factor of $1/ \det ( \vec{v}_1,\ldots,\vec{v}_d
  )$. After this transformation, $S'$ is a cube with side length
  two.
\end{proof}

The bound in Lemma~\ref{lem:intersection} depends on the determinant
formed by the $d$ direction vectors, which corresponds to the volume
of the $(d-1)$-simplex spanned by them. In order to obtain a good
upper bound, we must guarantee that this volume does not become too
small. Ensuring this reduces to the \emph{Heilbronn Problem}: Given
$k\in\N$ and a compact region $P\subset\R^d$ of unit volume, how can
we select $k$ points from $P$ as to maximize the area of the smallest
$d$-simplex formed by these points? Heilbronn posed this question for
$d=2$, the natural generalization to higher dimension was studied by
Barequet~\cite{barequet2001lower} and
Lefmann~\cite{DBLP:journals/combinatorica/Lefmann03}. We use the
following simple explicit construction in $\R^2$ that goes back to
Erd\H{o}s and was generalized to higher dimension by Barequet.
\begin{lemma}[\cite{barequet2001lower,r-ph-51}]\label{lem:EasyHeilbronn}
  Given a prime $k$, let $P=\{p_0,\ldots,p_{k-1}\}\subset{[0,1]}^d$ with
  \[
  p_i=\frac{1}{k}\left(i,i^2\mathop{\mathrm{mod}} k,\ldots,
    i^{d}\mathop{\mathrm{mod}} k\right).
  \]
  Then the smallest $d$-simplex spanned by $d+1$ points from $P$ has
  volume at least $1/(d!k^d)$.
\end{lemma}
Assuming $k$ to be prime is not a restriction: If $k$ is not prime,
then by Bertrand's postulate there is a prime $k'\le 2k$. We can
compute $k'$ efficiently, for instance, in $O(k/\log\log k)$ time
using Atkin's sieve~\cite{ab-psbqf-04}. In order to obtain the desired
direction vectors we proceed as follows: Use
Lemma~\ref{lem:EasyHeilbronn} to generate $k$ points
$p_0,\ldots,p_{k-1}$ in $[0,1]^{d-1}$. Then lift the points to
$S^{d-1}\subset\R^d$ using the map
\[
f:(x_1,\ldots,x_{d-1})\mapsto
\frac{(x_1-\frac12,\ldots,x_{d-1}-\frac12,\frac12)}{||(x_1-\frac12,\ldots,x_{d-1}-\frac12,\frac12)||}
\]
and denote the resulting set of directions by
$D=\{\vec{v}_0,\ldots,\vec{v}_{k-1}\}$ with $\vec{v}_i=f(p_i)$.

\begin{restatable}{lemma}{lemDeterminant}\label{lem:determinant}
  For any $d$ vectors $\vec{v}_{i_1},\ldots,\vec{v}_{i_d}$ from $D$ we
  have $|\det(\vec{v}_{i_1},\ldots,\vec{v}_{i_d})|\ge
  2^{d-1}/((d-1)!d^{d/2} k^{d-1})$.
\end{restatable}
\begin{proof}
  Let $p_j=(x_{j,1},\ldots,x_{j,d-1})$, for $j\in\{0,\ldots,d\}$. Then 
  \begin{eqnarray*}
    |\det(\vec{v}_{i_1},\ldots,\vec{v}_{i_d})| &=&
    \left|\det\left(
        \begin{array}{ccc}
          x_{i_1,1}-\frac12   & \hdots & x_{i_d,1}-\frac12\\
          \vdots    &  \ddots & \vdots\\
          x_{i_1,d-1}-\frac12 & \hdots  & x_{i_d,d-1}-\frac12\\
          \frac12         & \hdots  & \frac12 
        \end{array} 
      \right)\right| \prod_{j=1}^d
    \frac{1}{||(x_{i_j,1}-\frac12,\ldots,x_{i_j,d-1}-\frac12,\frac12)||}\\
    &=& \frac12\left|\det\left(
        \begin{array}{ccc}
          x_{i_1,1}   & \hdots & x_{i_d,1}\\
          \vdots    &  \ddots & \vdots\\
          x_{i_1,d-1} & \hdots  & x_{i_d,d-1}\\
          1        & \hdots  & 1
        \end{array} 
      \right)\right| \prod_{j=1}^d
    \frac{1}{||(x_{i_j,1}-\frac12,\ldots,x_{i_j,d-1}-\frac12,\frac12)||},
  \end{eqnarray*}
  where the determinant on the previous line describes the volume of
  the $(d-1)$-simplex spanned by $p_{i_1},\ldots,p_{i_d}$. According to
  Lemma~\ref{lem:EasyHeilbronn} this determinant is bounded by
  $1/((d-1)!k^{d-1})$ from below. Also note that all $p_i$ are in the unit
  cube and so all coordinates of the vector
  $(x_{i_j,1}-\frac12,\ldots,x_{i_j,d-1}-\frac12,\frac12)$ are between
  $-1/2$ and $1/2$. It follows that
  \[
  |\det(\vec{v}_{i_1},\ldots,\vec{v}_{i_d})| \ge
  \frac{1}{2(d-1)!k^{d-1}}\prod_{j=1}^d\frac{1}{\sqrt{d/4}} =
  \frac{2^{d-1}}{(d-1)!d^{d/2}k^{d-1}}. \qedhere
  \]
\end{proof}

\noindent
We are now ready to prove Theorem~\ref{thm-approxPlane}.
\begin{proof}
  The algorithm goes as follows. Compute directions
  $\vec{v}_1,\ldots,\vec{v}_k$ as in Lemma~\ref{lem:determinant}. For
  each $i\in\{1,\ldots,k\}$ consider the sequence of center points of
  the disks in $\D$, sorted according to direction $\vec{v}_i$, and
  denote by $S_i$ the middle $b$ points in this order (rank $(n-b)/2$
  up to $(n+b)/2$). We can bound
  \[
  kb=\sum_{i=1}^k |S_i| \leq (d-1)n + \sum_{i_1<\ldots <i_d}|S_{i_1}\cap
  \ldots \cap S_{i_d}|,
  \]
  noting that a point that is contained in at most $d-1$ sets $S_i$ is
  counted $d-1$ times on the right hand side, whereas a point that is
  contained in $a\ge d$ sets is counted $d-1+\binom{a}{d}\ge a$ times.

  Denote by $w_i$ the width of $S_i$ in direction $\vec{v}_i$ (which
  is the spread of $\vec{v}_i$). We claim that $w_i\ge t$, for some
  $i\in\{1,\ldots,k\}$. 

  For the purpose of contradiction assume $w_i < t$, for all
  $i\in\{1,\ldots,k\}$. Together with Lemma~\ref{lem:intersection} and
  Lemma~\ref{lem:determinant} we get
  \begin{eqnarray*}
    kb &=& \sum_{i=1}^k |S_i|\le (d-1)n + \sum_{i_1<\ldots <i_d} \frac{2^d
      w_{i_1}\ldots w_{i_d}}{| \det \left(
        \vec{v}_{i_1},\ldots,\vec{v}_{i_d} \right)| V_d }\\
    &<& (d-1)n +
    \sum_{i_1<\ldots <i_d} \frac{2^dt^d}{V_d}\frac{(d-1)!d^{d/2}k^{d-1}}{2^{d-1}}\\
    &=& (d-1)n + \binom{k}{d}\frac{2t^d (d-1)!d^{d/2}k^{d-1}}{V_d}\\
    &\le& (d-1)n + \frac{2d^{(d-2)/2}}{V_d}t^dk^{2d-1}.
  \end{eqnarray*}
  In combination with Condition~(\ref{cond:t1}) we get
  \[
  dn \le kb < (d-1)n + \frac{2d^{(d-2)/2}}{V_d}t^dk^{2d-1}
  \]
  and so
  \[
  t^d > \frac{V_d}{2d^{(d-2)/2}}\frac{n}{k^{2d-1}},
  \]
  in contradiction to the definition of $t$ in
  Condition~(\ref{cond:t2}). Therefore, our assumption $w_i < t$, for
  all $i\in\{1,\ldots,k\}$, was wrong and there is some $w_j\ge t$.

  Using Lemma~\ref{lem:Selectpoint} on the set $S_j$ projected to a
  line in direction $\vec{v}_j$ we obtain a hyperplane $\Hy$
  orthogonal to $\vec{v}_j$ that intersects at most $2b/(w_j-2)\le
  2b/(t-2)$ balls from $\D$. By Lemma~\ref{lem:Selectpoint} the
  hyperplane $\Hy$ has distance greater than one to any disk in $\D$
  whose center is not in $S_j$, and so $\Hy$ is the desired separator.
  
  Regarding the runtime bound, as stated above we can compute the
  spread of any direction in $O(n)$ time, which yields $O(kn)$ time
  for $k$ directions. The second step of finding $\Hy$ can be done in
  $O(b)=O(n)$ time by Lemma~\ref{lem:Selectpoint}. Therefore the
  overall runtime is $O(kn)$.
\end{proof}

% XXXMH: Why is this interesting?\\
% Note that the proof \emph{does not} use that the hyperplane $\Hy$ is
% halving. We only use that we have some interval for the choice of
% $\Hy$.

\section{A deterministic linear time algorithm in the plane}
\label{sec:computing_a_halving_line_with_few_intersections}

In this section we describe a deterministic linear time algorithm to
construct a halving line $\ell$ for a given set $\D$ of $n$ disks in
the plane. The line $\ell$ bisects $\D$ perfectly (at most $n/2$
centers lie on either side) and it intersects at most $O(n^c)$ disks,
where $c$ may be chosen arbitrarily close to $5/6$. We may assume that
$n$ is odd: If $n$ is even, remove one arbitrary disk and observe that
any halving line for the resulting set of disks is also a halving line
for the original set. As our algorithm works in the dual arrangement,
we first briefly review this duality and how it applies to line-disk
intersections.

\paragraph{Point-line duality.} The standard duality transform maps a
point $p=(p_x,p_y)$ to the line $p^* \colon y=p_x x - p_y$ and a
non-vertical line $g:y=mx+b$ to the point $g^*=(m,-b)$. This
transformation is both incidence preserving ($p\in g\iff g^*\in p^*$)
and order preserving ($p$ is above $g$ $\iff$ $g^*$ is above
$p^*$). Given a set $P$ of points in the plane, the dual arrangement
$\mathcal{A}(P^*)$ is defined by the lines in $P^*=\{p^* \mid p\in
P\}$. In order to avoid parallel lines we assume that no two points in
$P$ have the same $x$-coordinate (which can be achieved by an
infinitesimal rotation of the plane).

A halving line $\ell$ for $P$ corresponds to a point $\ell^*$ in the
dual arrangement that has no more than half of the lines from $P^*$
above it and no more than half of the lines below it. The set of these
points is referred to as the \emph{median level} of the arrangement
induced by $P^*$. Since $n$ is odd, for any $x$-coordinate there is
exactly one such point, and so we can regard the median level as a
function from $\R$ to $\R$. The following lemma characterizes
line-disk intersections in the dual plane. 
\begin{restatable}{lemma}{dualLineDisk}%[Dual line-disk intersection]
  \label{lem:diskintersection}
  Let $\ell:y = mx + b$ be a non-vertical line and let $p$ denote the
  center of a unit disk $D$. Then $D$ intersects $\ell$ if and only if
  the line $p^*$ intersects the vertical segment $s = [(m,-b
  -\sqrt{m^2 + 1}), (m,-b + \sqrt{m^2 + 1})]$.
\end{restatable}
\begin{proof}
  Consider $\ell$ and the two lines $\ell_a$ (above) and $\ell_b$
  (below) at distance $1$ from $\ell$ (\figurename~\ref{fig:dual}).
  Then $D$ intersects $\ell$ if and only if $p$ is below $\ell_a$ and
  above $\ell_b$. Equivalently, in the dual, $D$ intersects $\ell$ if
  and only if $p^*$ intersects the vertical line segment
  $\ell_a^*\ell_b^*$ at $x=m$. It remains to calculate the
  $y$-coordinates of the endpoints of $\ell_a^*\ell_b^*$.
  \begin{figure}[htbp]
    \centering
    \includegraphics{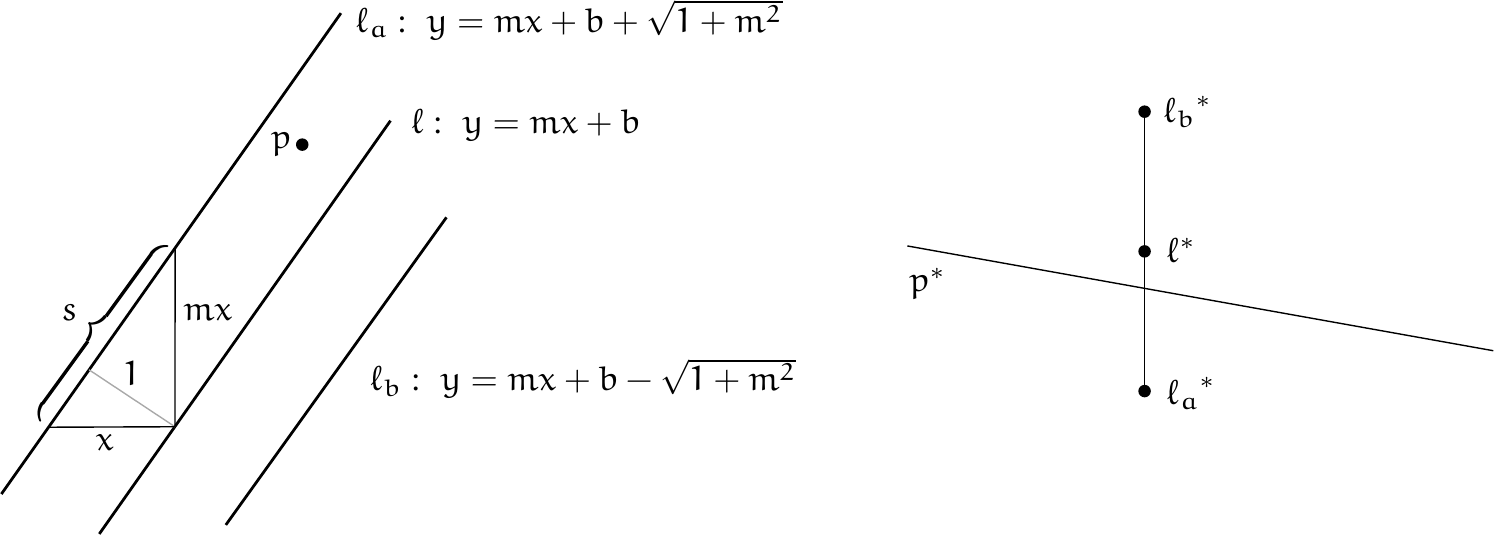}
    \caption{When does a line $\ell$ intersect a unit disk centered at
      $p$?\label{fig:dual}}
  \end{figure}

  Consider a right-angled triangle $T$ for which one side determines
  the horizontal distance $x$ and another side determines the vertical
  distance $mx$ between $\ell$ and $\ell_a$. Denote the length of the
  third side of $T$ by $s$. Then the area of $T$ is
  $\frac{1}{2}s=\frac{1}{2}mx^2$. By Pythagoras we have
  $s^2=x^2(1+m^2)$, which together yields $1+m^2={(mx)}^2$, and so
  $mx=\sqrt{1+m^2}$.
\end{proof}

\noindent
If we view Lemma~\ref{lem:diskintersection} from the perspective of a
unit disk $D$ with center $p$, then the set of lines that intersect
$D$ dualizes to the set of points $(x,y)$ whose vertical distance to
$p^*$ is at most $\sqrt{1+x^2}$. We call this closed region of points
the (dual) \emph{$1$-tube} of $D$ (figurename~\ref{fig:1tube}). Note
that the function $\sqrt{1 + x^2}$ is strictly convex and so the
$1$-tube is bounded by a strictly convex function from above and by a
strictly concave function from below.
\begin{figure}[htbp]
  \centering
  \includegraphics[scale=1.3,trim=20 20 20 20]{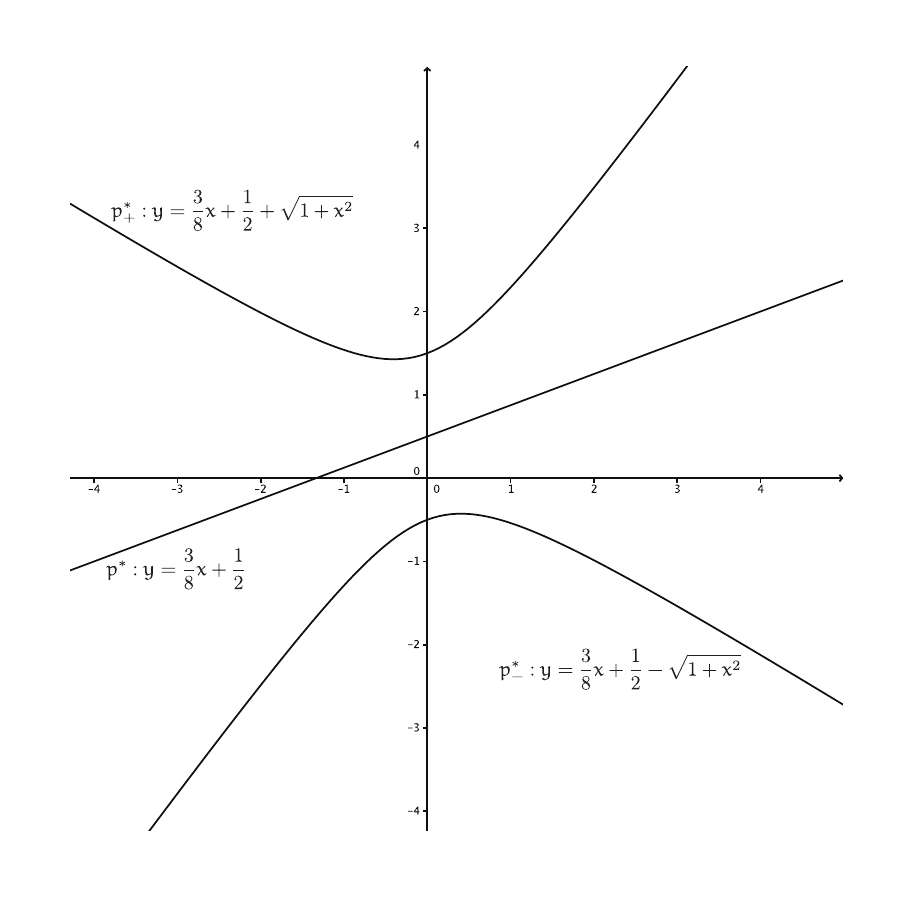}
  \caption{The $1$-tube of the disk centered at $p=(3/8,-1/2)$. It is
    bounded from below by the function $p^*_{-}=p^*-\sqrt{1+x^2}$ and
    from above by $p^*_{+}=p^*+\sqrt{1+x^2}$.}
  \label{fig:1tube}
\end{figure}

\paragraph{Overview of the algorithm.} The algorithm works in the dual
arrangement and follows the prune and search paradigm. At the
beginning we consider all potential halving lines, but subsequently
narrow the range of potential slopes for the desired halving
line. Recall that in the dual a halving line appears as a point on the
median level, whose x-coordinate corresponds to the slope of the
(primal) line.

The successive narrowing of the range of slopes under consideration is
made explicit by a parameter $S$, denoting the closed region bounded
by at most two vertical lines. Such a region we call a \textit{slab}.
A slab $S=\{(x,y)\in\R^2\colon \ell\le x\le r\}$ we denote by
$S={<}\ell,r{>}$. The distance $r-\ell$ between the two bounding
vertical lines is the \textit{width} of $S$. By Alon et
al.~\cite{DBLP:journals/dcg/AlonKP89} we may start with $S={<}0,1{>}$
as an initial slab, that is, there is always a halving line that
intersects few disks and whose slope is between zero and one.

Crucial for the linear runtime bound is that a constant fraction of
all lines from $L$ be discarded after each iteration. However, by
discarding some lines also our level of interest---which is the median
level of the original set of lines---changes. Therefore this level
also appears as a parameter of the algorithm. We denote this parameter
by $\lambda\in\{1,2,\ldots,|L|\}$. Initially $\lambda=\lceil
n/2\rceil$.

We first describe a single iteration of the algorithm, then prove some
bounds for the parameters, and finally present the analysis of the
whole algorithm.

\paragraph{A single iteration.} At the beginning of each iteration we
have a set $L$ of $n$ lines, a slab $S={<}\ell,r{>}$ of width
$w=r-\ell$, and a level parameter $\lambda$. Our goal is to find a
constant fraction of lines from $L$ that can be discarded. The outline
of an iteration step is as follows.
\begin{enumerate}
\item Divide $S$ in constantly many slabs $S_1,\ldots,S_m$, such that each
  contains at most $\alpha\binom{n}{2}$ many vertices of the arrangement
  $\mathcal{A}(L)$, for some appropriate constant $0<\alpha<1$. We define
  $S_i={<}\ell_i,r_i{>}$ and $w_i=r_i-\ell_i$.
\item For each slab $S_i$, construct a \emph{trapezoid} $T_i\subseteq
  S_i$ such that $T_i$ contains the $\lambda$-level of
  $\mathcal{A}(L)$ within $S_i$ and at most half of the lines from $L$
  intersect $T_i$.
\item\label{step:count} For each trapezoid $T_i$, define its
  \emph{$1$-tube} $\tau_i\supset T_i$ as follows: Consider the two
  lines $a_i$ and $b_i$ passing through the segment bounding $T_i$
  from above and below, respectively; then $\tau_i$ is defined as the
  closed subset of $S_i$ that is bounded by the upper boundary of the
  $1$-tube of $a_i$ from above and by the lower boundary of the
  $1$-tube of $b_i$ from below (\figurename~\ref{fig:trapetube}).

  For each slab $S_i$ and some parameter $\gamma\in(0,1/2)$, define
  the \emph{$\gamma$-core} $\mathrm{C}_{\gamma}$ of $S_i$ to be the
  central $(1-2\gamma)$-section of $S_i$, that is,
  $\mathrm{C}_{\gamma}(S_i)={<}\ell_i+\gamma w_i ,r_i-\gamma w_i{>}$.
 
  For each slab $S_i$, count the number $n_i$ of lines that intersect
  $\tau_i$ within $\mathrm{C}_\gamma(S_i)$.
\item\label{step:discard} Select (in a way to be described) one of the
  slabs $\mathrm{C}_\gamma(S_i)$ to continue the search with. Discard
  all lines from $L$ that do not intersect $\tau_i$ within
  $\mathrm{C}_\gamma(S_i)$ and adjust $\lambda$ accordingly (decrease
  by the number of lines discarded that are below $\tau_i$).
\end{enumerate}
\begin{figure}[htbp]
  \centering\includegraphics{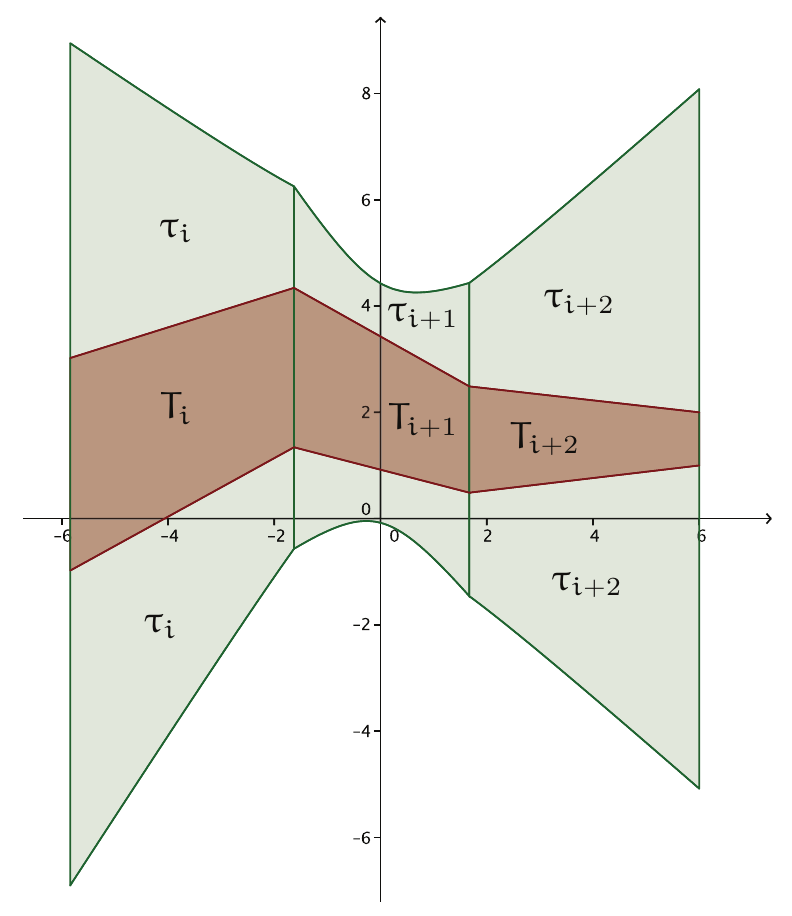}
  \caption{Three consecutive trapezoids $T_i,T_{i+1},T_{i+2}$ (shown
    in dark red) and their encompassing $1$-tubes
    $\tau_i,\tau_{i+1},\tau_{i+2}$, respectively (shown in light
    green). In our algorithm, all trapezoids are contained in
    ${<}0,1{>}$; in this figure they are spread out further so as to
    emphasize the convex/concave boundary of the $1$-tubes (which would be hard to
    recognize, otherwise).
    \label{fig:trapetube}}
\end{figure}

Observe first that discarding lines as described in
Step~\ref{step:discard} is justified: A line $\ell\in L$ that does not
intersect $\tau_i$ within $\mathrm{C}_{\gamma}(S_i)$ by
Lemma~\ref{lem:diskintersection} corresponds to a unit disk centered
at $\ell^*$ that within $\mathrm{C}_{\gamma}(S_i)$ is not intersected
by any line whose dual point lies on the $\lambda$-level of
$\mathcal{A}(L)$.

Next we will detail the steps listed above and analyze their
runtime. For the first two steps we apply the machinery due to Lo et
al.~\cite{DBLP:journals/dcg/LoMS94}. The first step can be handled in
linear time using the following lemma, which follows from Lemma~3.3 of
Lo et al.\ with $\alpha=1/32$.
\begin{lemma}[\cite{DBLP:journals/dcg/LoMS94}]
  Let $L$ be a set of $n$ lines in the plane in general
  position\footnote{Any two intersect in exactly one point.} and let
  $S$ be a slab. In $O(n)$ time $S$ can be subdivided into subslabs
  $S_1, S_2,\ldots,S_m\subset S$ (for some constant $m\le 64$), such
  that each $S_i$ contains at most $\frac{1}{32}\binom{n}{2}$ of the
  $\binom{n}{2}$ vertices of $\mathcal{A}(L)$.
\end{lemma}
The trapezoids mentioned in the second step can be computed as
follows. For $S_i={<}\ell_i,r_i{>}$, let the upper left (right) corner
of $T_i$ be defined by the $(\lambda+n/8)$-level of $\mathcal{A}(L)$
at $x=\ell_i$ ($x=r_i$). Analogously, the lower corners of $T_i$ are
defined by the $(\lambda-n/8)$-level of $\mathcal{A}(L)$ at $x=\ell_i$
($x=r_i$). Then Lemma~3.5 from the paper by Lo et al.\ (with
$\delta=1/8$) gives the following:
%\todo{Reviewer 1 suggests we restate
%  Lemma~3.5 here.}TM: Maybe for the journal version. It is so easy, that we
%  could even restate the proof, but I think only for the journal version.
%
\begin{lemma}[\cite{DBLP:journals/dcg/LoMS94}]
  \label{lem:bound_trapezoid}
  The trapezoid $T_i$ contains the $\lambda$-level of $\mathcal{A}(L)$
  within $S_i$ and at most half of the lines from $L$ intersect $T_i$.
\end{lemma}
All these trapezoids can be constructed in a brute-force manner in
$O(n)$ time (recall that $m$ is constant). This completes the first
two steps: we have computed (in linear time) a subdivision of our
initial slab $S$ into $m\le 64$ subslabs $S_i$, each of which contains
a trapezoid $T_i$ that contains the $\lambda$-level of
$\mathcal{A}(L)$ within $S_i$ and at most half of the lines from $L$
intersect $T_i$.

Regarding Step~\ref{step:count}, note that testing whether a given
line intersects $\tau_i$ is a geometric predicate of constant
algebraic degree and, therefore, can be done in constant time. Hence
this step can be executed in a straightforward manner in $O(mn)=O(n)$
time. It remains to argue how to select an appropriate slab to
continue with in Step~\ref{step:discard}. It turns out that not only
the number of lines matters, but it is also important to ensure that
the width of the slab does not become too small. The following lemma
gives a precise account for the bounds we are after.
\begin{lemma}\label{lem:iteration}
  For any $0<\varepsilon<1/2$ and $0<\gamma<1/2$ there exist an
  integer $n'>0$ and constants $m\le 64$ and
  $c={(8m/\gamma\epsilon)}^2$ such that for any $n\ge n'$ the
  following statement holds.

  Given a set $L$ of $n$ lines, an integer $\lambda\in\{1,\ldots,n\}$,
  and a slab $S\subseteq{<}0,1{>}$ of width $w\ge c\log(n)/n$, there
  exist a set $L'\subset L$ of at most $(\frac{1}{2}+\varepsilon)n$
  lines and a slab $S'$ of width $\ge(1-2\gamma)w/m$ such that inside
  $S'$ the $\lambda$-level of $\mathcal{A}(L)$ does not intersect any
  line in $L\setminus L'$.
\end{lemma}
\paragraph{Analysis of the algorithm.}
Let us postpone the proof of Lemma~\ref{lem:iteration} for now and
first complete the overall analysis of the algorithm. Denote by $n_t$
the number of lines and by $w_t$ the width of the current slab after
$t$ iterations. We have $n_0 = n$ and $w_0=1$. By
Lemma~\ref{lem:iteration} we have
\[
n_t\le\left(\frac12+\varepsilon\right)^t n
\quad\mbox{and}\quad
w_t \geq \left( \frac{1-2\gamma}{m} \right)^t ,
\]
as long as $w\ge c\log(n)/n$. After some number of iterations, either
we are left with a constant number of lines or a slab of width
$w<c\log(n)/n$. As in the first case we can finish by brute force, let
us concentrate on the second case. Suppose $t^*$ is the smallest index
for which $w_{t^*}<c\log(n)/n$. The following inequalities are
equivalent:
\begin{eqnarray*}
  \left( \frac{1-2\gamma}{m} \right)^{t^*} & < &
  \left(\frac{8m}{\gamma\varepsilon}\right)^2\cdot\frac{\log n}{n} \\
  -t^*\log\left(\frac{m}{1-2\gamma}\right) & < &
  2\log\left(\frac{8m}{\gamma\epsilon}\right)+\log\log n-\log n\\
  t^* & > &
  \frac{\log n - 2\log\left(\frac{8m}{\gamma\epsilon}\right) - \log\log n}{\log\left(\frac{m}{1-2\gamma}\right)}.\\
\end{eqnarray*}
Since $\gamma$, $\epsilon$ and $m$ are all constant, the last inequality
implies that for any constant $0<\varepsilon'<1$ we have
\[
t^* > \log n \cdot \frac{(1-\varepsilon')}{\log(\frac{m}{1-2\gamma})} \,,
\]
for sufficiently large $n$ (depending on $\varepsilon'$). Hence the
number of lines to be considered after $t^*$ iterations is
% short version of equations below
% \[
% n_{t^*} \le n\cdot\left( \frac{1}{2}+\varepsilon \right)^{t^*} <
% n\cdot\left(\frac{1}{2}+\varepsilon \right)^{ \log n
%   \frac{1-\varepsilon'}{\log(\frac{m}{1-2\gamma})}} = n^{1+ \log\left(
%     \frac{1}{2}+\varepsilon
%   \right)\frac{1-\varepsilon'}{\log(\frac{m}{1-2\gamma})}} \leq
% n^{\frac56+\delta} \;,
% \]
\begin{eqnarray*}
  n_{t^*} & \le & \left( \frac{1}{2}+\varepsilon \right)^{t^*}\cdot n\\
  & < & \left(\frac{1}{2}+\varepsilon
  \right)^{ \log n \frac{1-\varepsilon'}{\log(\frac{m}{1-2\gamma})}}\cdot n\\
  &= & n^{ \log\left( \frac{1}{2}+\varepsilon
  \right)\frac{1-\varepsilon'}{\log(\frac{m}{1-2\gamma})}}\cdot n\\  
  &= & n^{ \log\left( \frac{1}{2}+\varepsilon
  \right)\frac{1-\varepsilon'}{\log(\frac{m}{1-2\gamma})} + 1}\\ 
  & \leq & n^{\frac56+\delta}\\
\end{eqnarray*}
where the last inequality uses $m\leq 64$ (and hence $\log m\leq 6$)
and where $\delta>0$ can be made arbitrarily small by choosing
$\varepsilon$, $\varepsilon'$ and $\gamma$ to be correspondingly
small.

So after at most $t^*=\Theta(\log n)$ iterations we are left with a
slab $S$ and $O(n^{\frac56+\delta})$ lines. All lines that have been
discarded do not intersect the $1$-tube of the level that corresponds
to the original median level. Therefore any point on this level within
$S$ corresponds to a halving line for the original set of disks that
intersects $o(n)$ of the disks. Such a point can easily be found in a
brute force manner in $o(n)$ time.

Denote by $R(n)$ the runtime of the algorithm for $n$ disks. Each
iteration can be handled in time linear in the number of lines/disks
remaining and so
\[
R(n)\le\sum_{t=0}^{t^*} cn_t\le cn\sum_{t=0}^{t^*}
\left(\frac{1}{2}+\varepsilon\right)^t< \frac{2c}{1-2\varepsilon}n =
O(n) ,
\]
for some constant $c$. This proves Theorem~\ref{thm:2dlinear}.

\paragraph{Proof of Lemma~\ref{lem:iteration}.}
It remains to prove that we can select a constant fraction of lines to
be discarded in each iteration while at the same time the width of the
current slab does not shrink too much. To begin with we need a slab
whose $1$-tube is not intersected by too many lines. To show that such
a slab exists, we use an averaging argument: While a single $1$-tube
$\tau_i$ may be intersected by all lines from $L$, on average the
number of intersecting lines per slab is sublinear. To this end we
define a function $g$ by setting $g(x)$ to be the number of lines that
intersect $\bigcup_{i=1}^m \tau_i$ at $x\in(\ell,r)$. The following
lemma provides an upper bound on the average number of such lines.
\begin{restatable}{lemma}{lemIntegral}\label{lem:integral}
  For a slab $S={<}\ell,r{>}\subseteq{<}0,1{>}$ of width $w=r-\ell$,
  there is some constant $c\le 4$ such that
  \[
  \int_{\ell}^r g(x)\,\mathrm{d}x \le c\sqrt{nw \log(nw)}\,,
  \]
  if $nw$ is sufficiently large.
\end{restatable}
\begin{proof}
  We follow the approach of
  Alon~et~al.~\cite{DBLP:journals/dcg/AlonKP89} but are more specific
  about some technical details. We define $x_i := \ell + iw/k$, for $i
  = 0, \ldots, k-1$ and some parameter $k$ to be specified later and
  consider the function
  \[
  f(x) := \sum_{i=0}^{k-1} g(x+x_i)
  \]
  over the domain $[0,w/k]$. Clearly, we have
  \[
  \int_0^{w/k} f(x)\,\mathrm{d}x = \int_{\ell}^r g(x)\,\mathrm{d}x\,.
  \]
  Next we bound $f(x)$ for some arbitrary but fixed $x$. To this end,
  we move back to the primal setting and consider the set $H$ of
  halving lines $h_i$ with slope $x_i+x$, for
  $i\in\{0,\ldots,k-1\}$. Let $D_i$ denote the set of disks from $\D$
  that intersect $h_i$. The value of $f$ is the number of pairs
  $(d,h)\in D\times H$ where $d\cap h\ne\emptyset$. A (generous) upper
  bound for this quantity is provided by
  \[
  n + \sum_{i<j}|D_i \cap D_j|\,,
  \]
  where the first term counts every disk that intersects only one line
  and the second term counts every disk that is intersected by at
  least two lines. (In this way, a disk that is intersected by $c$
  lines is counted $1+\binom{c}{2}$ times.)

  Let $\vec{v}_i={(-x_i-x,1)}^T/\sqrt{{(x_i+x)}^2+1}$ be the (unit)
  normal vector to $h_i$. By Lemma~\ref{lem:intersection} (where
  $d=2$, $w_1=w_2=2$, $\vec{v}_1=\vec{v}_i$, and
  $\vec{v}_2=\vec{v}_j$) we have (using $x+x_i\leq 1$)
  \[
  |D_i\cap D_j|\le\frac{16}{\pi\,|\det(\vec{v}_1,\vec{v}_2)|}=
  \frac{16\sqrt{(x_i+x)^2+1}\sqrt{(x_j+x)^2+1}}{\pi\,|x_i-x_j|}\le
  \frac{32}{\pi\,|x_i-x_j|}=\frac{32k}{\pi w\,|i-j|}
  \]
  and therefore
  \[
  n + \sum_{i<j}|D_i \cap D_j| \leq n + \frac{32k}{\pi w}\sum_{i<j}
  \frac{1}{j-i}\,.
  \]
  The sum can be bounded using
  \[
  \sum_{i<j} \frac{1}{j-i} = \sum_{a=1}^{k-1}\frac{k-a}{a} =
  kH_{k-1}-(k-1)<1+k\ln(k)\,,
  \]
  where the last inequality uses the well-known bound $H_n<1+\ln(n)$
  for the harmonic number. We started out by fixing a particular $x$,
  but the derived bound holds for any arbitrary $x$. Altogether we
  obtain
  \[
  f(x) < n + \frac{32k}{\pi w}(1+k\ln(k))=n + \frac{32}{\pi
    w}(k+k^2\ln(k))
  \]
  and so
  \[
  \int_{\ell}^r g(x)\,\mathrm{d}x = \int_0^{w/k} f(x)\,\mathrm{d}x <
  \frac{nw}{k} + \frac{32}{\pi}(1+k\ln(k))\,.
  \]
  %
  % In[1]:=f[x_] := ReplaceAll[
  % (x/k + (32/Pi) (1 + k Log[k]))
  %  -
  %  (x/Ceiling[
  %      k] + (32/Pi) (1 + Ceiling[k] Log[Ceiling[k]])), {k -> (Sqrt[
  %       Pi] Sqrt[x/Log[x]])/4}]
  %
  % In[2]:=N[f[512]]
  % Out[2]=0.0061684
  %
  % In[3]:=N[f[511]]
  % Out[3]=-0.0185759
  %
  % In[4]:=Plot[f[x], {x, 500, 550}]
  %
  Setting $k=\left\lceil\sqrt{\pi nw/(16\ln(nw))}\right\rceil$ in the
  previous expression and omitting the ceilings (it can be verified
  that this only increases the value of the expression, provided
  $nw\ge 512$) yields
  \[
  \int_{\ell}^r g(x)\,\mathrm{d}x <
  \frac{4}{\sqrt{\pi}}\sqrt{nw\ln(nw)} +\frac{32}{\pi}
  +\frac{8}{\sqrt{\pi}}\sqrt{\frac{nw}{\ln(nw)}}\ln\left(\sqrt{\frac{\pi
        nw}{16\ln(nw)}}\right)\,,
  \]
  which---noting that $\sqrt{\pi x/(16\ln x)}<\sqrt{x}$, for $x\ge
  1$---is upper bounded by
  \[
  \frac{4}{\sqrt{\pi}}\sqrt{nw\ln(nw)} +\frac{32}{\pi}
  +\frac{8}{\sqrt{\pi}}\sqrt{\frac{nw}{\ln(nw)}}\ln(\sqrt{nw})=
  \frac{8}{\sqrt{\pi}}\sqrt{nw\ln(nw)} +\frac{32}{\pi}.
  \]
  It can be checked that the last expression is upper bounded by
  $4\sqrt{nw\log_2(nw)}$, for $nw\ge 226$.
\end{proof}

By the pigeonhole principle, the integral is small for most subslabs. But
bounding the integral is not sufficient to bound the number of lines that
intersect the $1$-tube, because lines that do so for a very short interval only
do not contribute much to the integral. To account for such lines we restrict
our focus to the $\gamma$-core of the slabs instead. For a slab $S_i$ let
$d_i(x)$ denote the number of lines that intersect $\tau_i\setminus T_i$ at
$x$, for $x\in(\ell_i,r_i)$. Clearly $d_i\le g$. Furthermore let
\[
\phi_{\gamma,i} = \max\,\{d_i(x)\colon
x\times\R\subset\mathrm{C}_{\gamma}(S_i)\} .
\]
\begin{proposition}\label{prop:linebound}
  The number of lines from $L$ that intersect $(\tau_i\setminus
  T_i)\cap \mathrm{C}_{\gamma}(S_i)$ is bounded by $2\phi_{\gamma,i}$,
  for any $i\in\{1,\ldots,m\}$ and $0<\gamma<1/2$.
\end{proposition}
\begin{proof}
  Let $\mathrm{C}_{\gamma}(S_i)={<}a_i,b_i{>}$ and consider a line
  $\ell$ that intersects $(\tau_i\setminus T_i)\cap
  \mathrm{C}_{\gamma}(S_i)$. Then $\ell$ intersects at most one
  boundary of $\tau_i$, say, the upper boundary $U$. As $U$ is
  strictly convex, the line $\ell$ intersects $\tau_i$ at $x=a_i$ or
  $x=b_i$ (possibly both). Therefore, the number of such lines is
  upper bounded by $d_i(a_i)+d_i(b_i)\le 2\phi_{\gamma,i}$.
\end{proof}

\begin{proposition}\label{prop:phibound}
  $\phi_{\gamma,i}w_i\le \gamma^{-1}\int_{\ell_i}^{r_i} g(x)\,\mathrm{d}x$
\end{proposition}
\begin{proof}
  Let $\mathrm{C}_{\gamma}(S_i)={<}a_i,b_i{>}$ and consider a line
  $\ell$ that is counted in $\phi_{\gamma,i}$, that is, $\ell$
  intersects $\tau_i\setminus T_i$ at some $x\in[a_i,b_i]$. By the
  proof of Proposition~\ref{prop:linebound}, we may assume that
  $x\in\{a_i,b_i\}$. Using the same argumentation, we may also assume
  that $\ell$ intersects $\tau_i$ at some $x'\in\{\ell_i,r_i\}$.
  Regardless of the combination of $x$ and $x'$, it follows that
  $\ell$ contributes to $d_i$---and thus to $g$---for at least a
  $\gamma$-fraction of the interval $[\ell_i,r_i]$.
\end{proof}

\begin{figure}[htbp]
  \centering
  \includegraphics{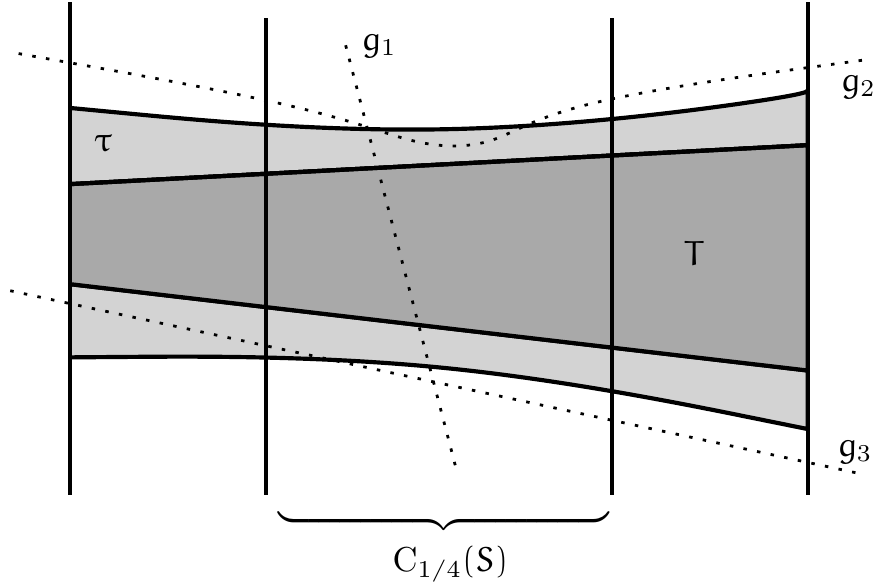}
  \caption{A slab $S$ with its core $\mathrm{C}_{1/4}(S)$ and a
    corresponding trapezoid $T$ with its $1$-tube $\tau$. The line
    $g_1$ intersects both $T$ and $\tau$, whereas $g_2$ and $g_3$
    intersect $\tau$ but not $T$. Every such line intersects $\tau$ at
    the boundary of the core, like $g_3$ does. An intersection pattern
    as depicted for $g_2$ is impossible for a straight line.}
  \label{fig:Core}
\end{figure}

\noindent Now we have all tools in place to complete the proof of
Lemma~\ref{lem:iteration}. Combining Proposition~\ref{prop:phibound} and
Lemma~\ref{lem:integral} yields
\[
\sum_{i=1}^m \phi_{\gamma,i}w_i \le 4\gamma^{-1}\sqrt{nw\log(nw)}\,.
\]

We claim that we can select any slab $S_j$ for which $w_j\ge w/m$ and
continue the search within $\mathrm{C}_{\gamma}(S_j)$. Such a slab
exists because there are $m$ slabs in total and $w=\sum_{i=1}^m
w_i$. We can then bound
\[
\phi_{\gamma,j}\frac{w}{m}\le\phi_{\gamma,j}w_j\le \sum_{i=1}^m
\phi_{\gamma,i} w_i \le 4\gamma^{-1}\sqrt{nw\log(nw)}
\]
and so
\[
\phi_{\gamma,j}\le 4\gamma^{-1}m\,\sqrt{\frac{n\log(nw)}{w}}\le
4\gamma^{-1}m\,\sqrt{\frac{n\log(n)}{w}}\,.
\]
The slab we continue to search in (the core $\mathrm{C}_{\gamma}(S_j)$
of $S_j$) has width at least $(1-2\gamma)w/m$. Lemma~\ref{lem:bound_trapezoid} and
Proposition~\ref{prop:linebound} bound the number $n_{\gamma,j}$ of
lines that intersect $\tau_j$ within $\mathrm{C}_{\gamma}(S_j)$ by
\[
n_{\gamma,j}\leq\frac{n}{2}+2\phi_{\gamma,j}\le\frac{n}{2}+\frac{8m}{\gamma}\,
\sqrt{\frac{n\log(n)}{w}}\,.
\]
Given any $0<\epsilon<1/2$ and $0<\gamma<1/2$, we have
$
n_{\gamma,j}\le\left(\frac{1}{2}+\varepsilon\right)n\,,
$
as long as
$
w\ge\left(\frac{8m}{\gamma\varepsilon}\right)^2\cdot\frac{\log
  n}{n}\,,
$
which is stated as an assumption. This completes the proof of
Lemma~\ref{lem:iteration}.

\section{Conclusions}
In this paper we studied the construction of separators for balls in
\emph{deterministic} linear time. The aim is to intersect as few balls
as possible while (approximately) bisecting the set of center
points. We presented essentially two ways to compute such seperators
with a sublinear number of intersections. The first algorithm is very
simple and straight-forward to implement (we gave all constants
explicitly), and obtains an arbitrarily good bisection in combination
with an asymptotically optimal number of intersections. The strength
of the second algorithm is to bisect the center points \emph{exactly},
but it works in the plane only.

Throughout the paper we assumed the balls to be disjoint, but we never
really used it. In fact, both algorithms work as long as we have some
density lower bound on the objects under consideration and some bound
on the size of the objects. This lower bound is implicitly given if
for instance the objects satisfy some fatness condition and are
disjoint. Also note that, in contrast to the continuous case, we do
not make use of the fact that the hyperplane to be constructed is
bisecting. Therefore it is easy to adapt the algorithm to, for
instance, have $n/3$ of the points on one side and $2n/3$ of the other
side of the hyperplane.

There are point sets for which the number of balls intersected by
\emph{every} halving hyperplane is $\Omega(n^{(d-1)/d})$. But already
for dimension three it is not clear if a halving plane with
$o(n^{3/4})$ intersections always exists ($O(n^{3/4})$ is not
difficult). In dimension two it is open if $o(\sqrt{n \log n})$ can be
achieved. So let us ask the following question: Is it true that for
every set of $n$ disjoint unit balls in $\R^d$ there exists a halving
hyperplane that intersects $O(n^{(d-1)/d})$ of the balls?

\paragraph{Acknowledgments.}
We want to thank Marek Elias, Ji\v{r}ka Matou\v{s}ek, Edgardo
Rold{\'a}n-Pensado and Zuzana Safernov{\'a} for interesting
discussions on the conjecture for higher dimensions and referring us
to related work.

%\todo{Reviewer 2 asks: Theorem 1 implies the existence of a separating
%  plane which crosses $O(n^{1-1/d})$ balls, and with Omega(n) balls to
%  each side. This is a nice number and a curious structural result
%  by itself. To the authors: Did this property appear anywhere before
%  in the space-partition literature? How about more general fat and/or
%  pairwise-wise disjoint objects?}
% TM: To the first sentence: YES, that is what the Theorem says.
% To the second question: I don't know. 
% And I think Jirka would have told us, if he knew.
% To the last question: Any objects with bounded density are fine.

%\todo{Reviewer 3 says: Finally, I think a discussion of randomized
%  algorithms is missing in my opinion. The techniques that they use
%  certainly does not require the full power of a solution for
%  Heilbronn's problem and thus it makes me wonder how well would a set
%  of random directions would do. }
%
%TM: The original motivation of this work is to de-randomize, so I haven't
%thought about it in higher dimensions. In the plane it is clear. In higher 
%dimensions: . . . I guess that a random direction will have high chances to 
%to be good w.h.p. But I think it is tedious to proof this. 
%We would need to bound the sum of triple intersections for $k$ randomly
%chosen hyperplanes. The bound wouldn't need to be good --- more or less any bound 
%that holds w.h.p. is sufficient. 
%If any of you feels like it: do it.

%%\bibliographystyle{plain}
%%\bibpunct{[}{]}{,}{n}{}{;} %

\bibliographystyle{mh-url} %
\bibliography{Lib}

\end{document}